\documentclass[
superscriptaddress, 
 reprint,
 amsmath,amssymb,
 aps, pra
]{revtex4-2}

\usepackage{graphicx}
\usepackage{hyperref}
\usepackage[linesnumbered, ruled, vlined, titlenotnumbered]{algorithm2e}
\usepackage{amsmath, amsfonts, amssymb, amsthm, color}

\usepackage[capitalise]{cleveref}
\crefname{algocf}{alg.}{algs.}
\Crefname{algocf}{Algorithm}{Algorithms}

\newtheorem{proposition}{Proposition}

\newtheorem{lemma}{Lemma}

\newcommand{\ket}[1]{|#1\rangle}

\DeclareMathOperator{\CliNR}{CliNR}

\newcommand{\Prob}{{\mathbb{P}}}
\newcommand{\Expectation}{{\mathbb{E}}}

\newcommand{\Z}{{\mathbb{Z}}}

\newcommand{\depth}{\Delta}

\usepackage{xspace}
\newcommand{\rsp}{{RSP}\xspace}
\newcommand{\rsv}{{RSV}\xspace}

\newcommand{\RSPV}{{RSP\&V}\xspace}

\begin{document}

\title{Advantage in distributed quantum computing with slow interconnects}

\author{Evan E. Dobbs}
\affiliation{IonQ Inc.}
\affiliation{Aalto University}

\author{Nicolas Delfosse}
\affiliation{IonQ Inc.}

\author{Aharon Brodutch}
\affiliation{IonQ Inc.}
\date{\today}

\begin{abstract}
The main bottleneck for distributed quantum computing is the rate at which entanglement is produced between quantum processing units (QPUs).
In this work, we prove that multiple QPUs connected through slow interconnects can outperform a monolithic architecture made with a single QPU.
We consider a distributed quantum computing model with the following assumptions:
(1) each QPU is linked to only two other QPUs,
(2) each link produces only one Bell pair at a time,
(3) the time to generate a Bell pair is $\tau_e$ times longer than the gate time.
We propose a distributed version of the CliNR partial error correction scheme respecting these constraints and we show through circuit level simulations that, even if the entanglement generation time $\tau_e$ is up to five times longer than the gate time, distributed CliNR can achieve simultaneously a lower logical error rate and a shorter depth than both the direct implementation and the monolithic CliNR implementation of random Clifford circuits.
In the asymptotic regime, we relax assumption (2) and we prove that links producing $O(t/\ln t)$ Bell pairs in parallel, where $t$ is the number of QPUs, is sufficient to avoid stalling distributed CliNR, independently of the number of qubits per QPU.
This demonstrates the potential of distributed CliNR for near-term multi-QPU devices.
Moreover, we envision a distributed quantum superiority experiment based on the conjugated Clifford circuits of Bouland, Fitzsimons and Koh implemented with distributed CliNR.
\end{abstract}

\maketitle

\section{Introduction}
\label{sec:Introduction}

A quantum computer possessing millions of qubits might be necessary to reach large-scale applications~\cite{reiher2017elucidating, beverland2022assessing, dalzell2023quantum, gidney2025factor, zhou2025resource}.
The most promising path to building such a colossal device is to design a distributed quantum computer made with multiple interconnected quantum processing units (QPUs).
Many monolithic quantum computers made with a single QPU are available today but quantum interconnect technologies are not as mature~\cite{awschalom2021development}.

Quantum interconnects generating entanglement between different QPUs have been demonstrated experimentally~\cite{Moehring2007EntanglementDistance, Crain2019HighSpeedDetection, Maunz2009HeraldedGate, Ritter2012Elementary, Slodicka2013SinglePhotonEntanglement, Casabone2013CavityEntanglement, Hucul2015ModularEntanglement, Jing2019Entanglement, Mirhosseini2020Superconducting, stephenson2020high, Pompili2021Multinode, Krutyanskiy2023Entanglement230m, Knaut2024Entanglement, o2024fast, Oreilly2024Thesis, saha2025high, Meesala2024Quantum, main2025distributed} but record entanglement generation speeds remain far slower than the operation time of a QPU.
For example, recent experiments produce a Bell pair over two ions in different traps every 4ms in average in~\cite{o2024fast, Oreilly2024Thesis} and every 5.5ms in~\cite{stephenson2020high} whereas operation times in a trapped ion quantum computer typically range between tens of microseconds and 1ms.
Moreover, the need for entanglement distillation might further increase the entanglement generation time~\cite{bennett1996concentrating, bennett1996purification, bennett1996mixed, deutsch1996quantum, fujii2009entanglement, jansen2022enumerating, kang2023trapped, ramette2024fault, pattison2025constant, sinclair2025fault, marqversen2025fault, gu2025constant}.
This raises the following questions.

\medskip
{\em 
Is there any advantage in connecting multiple QPUs with slow quantum interconnects?
}
\medskip

To get some insight in the amount of entanglement required to perform a distributed quantum computation, consider the following scenario. Take $Tn$ qubits, distributed equally over $T$ QPUs and assume that a stream of random two-qubit gates must be executed. Such a random gate is acting on qubits from the same QPUs with probability $\frac{n-1}{nT} \rightarrow \frac{1}{T}$.
Therefore, to guarantee that remote two-qubit gates never stall the computation, one may need to generate the entanglement they consume $T$ times faster than the gate time of a single QPU.
Moreover, these two-qubit gates connect arbitrary pairs of QPUs which may require a fully connected topology for the QPU network.
This is far beyond current experimental results.

Another insight comes from circuit knitting techniques~\cite{piveteau2023circuit, lowe2023fast, harrow2025optimal}, see also~\cite{bravyi2016trading}. 
If entanglement generation is very slow then only a small number of remote two-qubit gates can be executed and circuit knitting techniques can be used to replace these remote gates by a moderate amount of classical computation together with error mitigation techniques on disconnected devices.
However, circuit knitting is limited to very slow interconnect because its cost grows exponentially with the number of remote gates.

Overall, it is unclear whether QPUs connected with slow interconnects can perform any task significantly better than disconnected devices.
To investigate this question, we introduce a model for distributed quantum computing with multiple QPUs connected as follows:
(1) each QPU is linked to only two other QPUs,
(2) each link produces only one Bell pair at a time,
(3) the time to generate a Bell pair is $\tau_e$ times longer than the gate time.
It is fair to refer to the model we study in this paper as {\em slow interconnects} because 
(1') we only use two links per QPU, that is only one more link than the minimum required to form a connected cluster of QPUs,
(2') we do not use any parallel links to generate multiple Bell pairs at once, even when the number of qubits per QPU grows,
(3') we focus on the regime where $1 \leq \tau_e \leq 15$ which means that Bell pair generation time is up to 15 times longer than the gate time.
Even with these pessimistic assumptions about the entanglement production speed, our numerical simulations demonstrate an advantage with a distributed architecture.

We consider the CliNR scheme which is a partial error correction scheme designed to reduce the noise rate of Clifford circuits at the price of a smaller overhead than quantum error correction~\cite{delfosse2025low}.
The Clifford circuit we want to implement is split into subcircuits and each subcircuit is implemented via state injection.
The state injection consumes a resource state that can prepared and verified offline and which can be discarded and re-prepared if an error is detected during the verification.

In this work, we propose a distributed version of CliNR where the resource state corresponding to each subcircuit is prepared and verified on a different QPU.
This task can be accomplished in parallel without consuming any Bell states.
We design a state injection procedure that can be implemented in a circular network of QPUs respecting the assumptions (1), (2) and (3) above.
The parallelization of the resource state preparation and verification helps in reducing the depth of the CliNR implementation and the accumulation of noise on idle qubits during all these resource state preparations.
However, if entanglement generation is too slow, the subsequent state injection stages, which consume entanglement, might be stalled, canceling the previous gains in depth and noise.
The structural reason making distributed CliNR well suited to a distributed architecture with slow interconnects is that the state injection only consumes $n$ Bell pairs per link of the circular network.
Therefore, a large fraction of the required Bell pairs, if not all, can be prepared during the resource state preparation and verification, which avoids stalling the state injection.

Our numerical simulation proves that distributed CliNR can achieve simultaneously a lower logical error rate and a shorter depth than both the direct implementation and the monolithic CliNR implementation of random Clifford circuits.
This simulation is performed with a $85$-qubit random Clifford circuits distributed over four modules.

In the asymptotic regime we relax assumption (2) by allowing parallel links producing multiple Bell states simultaneously between two connected QPUs.
In this asymptotic regime, we prove that for Clifford circuits distributed in a uniform way (see \cref{sec:Theory result}) over $T$ QPUs, a circular network with $O(T/\ln T)$ parallel links per connection is sufficient for implementing distributed CliNR without any delay due to unavailable entanglement.

The rest of this paper is organized as follows. \cref{sec:Architecture model} describes our distributed quantum computing model. \cref{sec:Distributed CliNR} reviews $\CliNR$ and introduces distributed $\CliNR$. Our numerical results are presented in \cref{sec:simulation} and \cref{sec:Theory result} establishes asymptotic bounds.

\section{Monolithic and distributed architectures}
\label{sec:Architecture model}

The {\em monolithic architecture} describes a single QPU with fully connected and maximally parallel qubits. 
It is a register of qubits equipped with the following operations: preparation of a qubit in the state $\ket{0}$, arbitrary single-qubit rotations, two-qubit gates $CX$, $CY$ and $CZ$, and single-qubit Pauli measurements, where two-qubit gates are available between all pairs of qubits and any set of non-overlapping operations can be implemented simultaneously.

A {\em distributed architecture} includes multiple QPUs connected through  links producing entanglement.
In this work, we restrict ourselves to a circular topology for the QPU network because it requires a few links.
We consider $T$ QPUs, denoted $Q_k$, indexed by $k \in \Z_T$, where $Q_k$ is connected to $Q_{k-1}$ and $Q_{k+1}$ by a link.
Each QPU comprises of three different modules described below: an \emph {entanglement generation module}, a \emph{storage module} and a \emph{compute module} and two {\em  links} as represented in \cref{fig:qpu}.

\begin{figure}
    \centering   
            \includegraphics[width=0.7\linewidth]{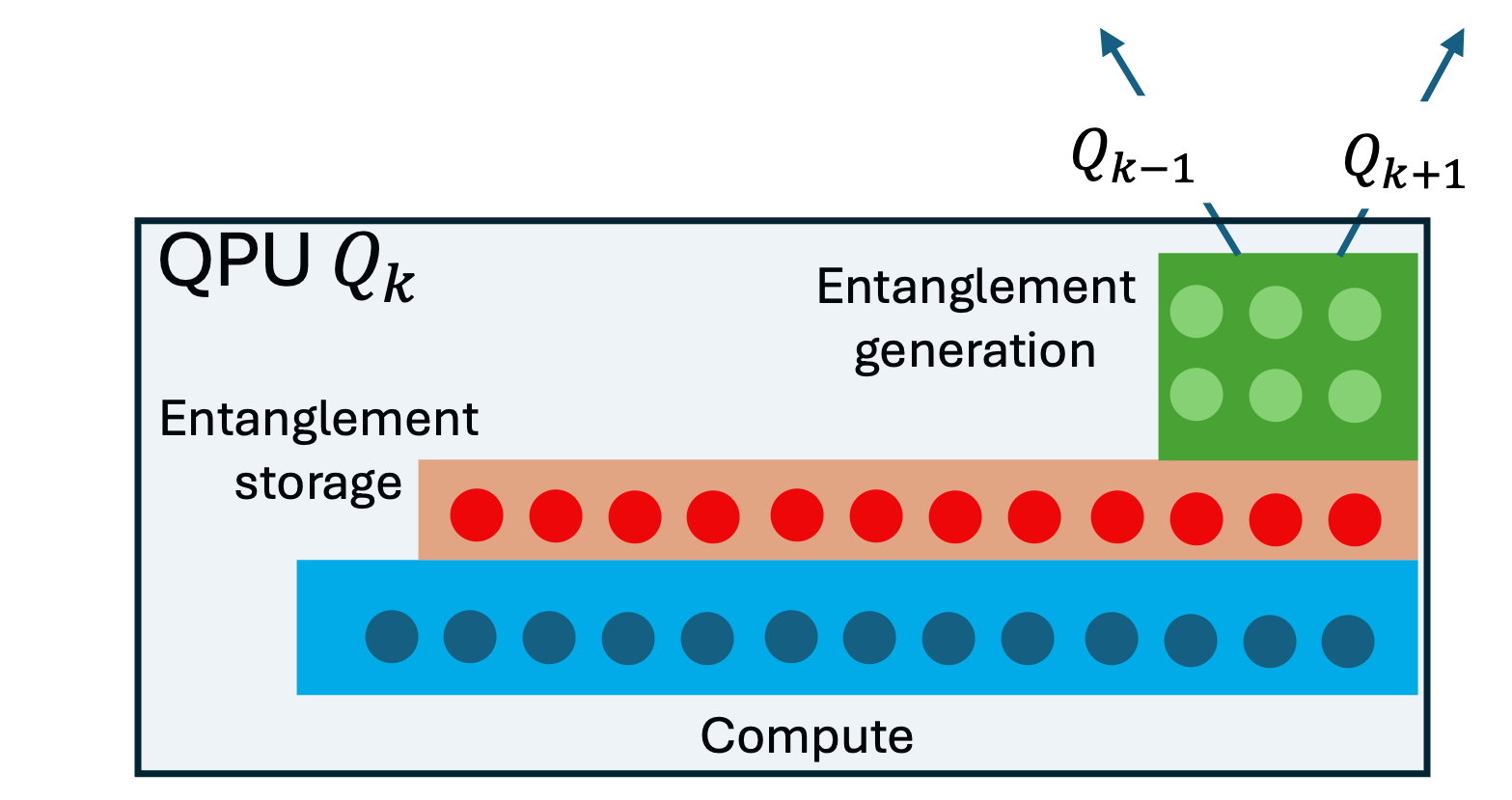}
    \caption{The three modules of a QPU $Q_k$ in a distributed architecture connected to two other QPU $Q_{k\pm1}$.
    }
    \label{fig:qpu}
\end{figure}

The {\em entanglement generation module} is used to generate Bell pairs between connected QPUs.
Depending on the quality of the links, these Bell pairs might be raw Bell states or distilled Bell states.
We omit the details of the entanglement generation and distillation protocols and we assume that one can produce Bell pairs with sufficiently high quality on any pair of connected modules in depth $\tau_e$.
The remaining noise in these Bell pairs is captured in our noise model in~\cref{sec:simulation} by increasing the noise rate of two-qubit gates consuming Bell pairs.
Moreover, we assume that one can produce Bell pairs supported on distinct pairs of modules simultaneously, but only one at a time over the same connection.
Therefore, the total number of Bell states produced in a circular QPU network with $T$ QPUs is limited to at most $T$ Bell pairs every $\tau_e$ time steps.

Once a Bell pair is generated, we assume that it is stored in the {\em storage module}, which consists of $n_s$ \emph{storage qubits}.
The storage qubits are used to store half of an EPR pair, the other half being stored in the storage module of a neighboring QPU.
We assume that they are long-lived~\cite{wang2020single}, which translates in negligible idle noise rate in the noise models considered in \cref{sec:simulation}. They are also simpler in the sense that they do not need to interact with each-other and have limited interactions with qubits in the compute module. 

The {\em compute module} has $n_c$ fully connected, maximally parallel qubits, and it is where the core of the computation takes place.
To keep the model simple, we assume that both the storage module and the compute module are maximally parallel. 

The storage qubits are used exclusively to implement a {\em remote gate}, that is, either a $CX$ gate or teleportation, between qubits in the compute modules of two different QPUs.
Each remote gate consumes one Bell pair shared between the storage modules of these two QPUs. For this purpose, the storage qubits can interact with compute qubits, and can be measured. 
If no such Bell pair is available in the storage modules, the remote gate is delayed until one is generated.

\section{Monolithic and Distributed CliNR}
\label{sec:Distributed CliNR}

This section reviews the CliNR scheme~\cite{delfosse2025low, tham2025optimized, brodutch2025recursive} and proposes a distributed CliNR implementation.

\subsection{Review of CliNR}
\label{subsec:review_of_clinr}

CliNR is a technique for reducing logical error rates in quantum circuits by performing  Clifford subcircuits offline and using stabilizer measurements to check for errors.  The subcircuits are injected to the main circuit if no errors are detected, otherwise the process is repeated until success.

\begin{figure}
    \centering
    (a) 
    
    \includegraphics[width=0.9\linewidth]{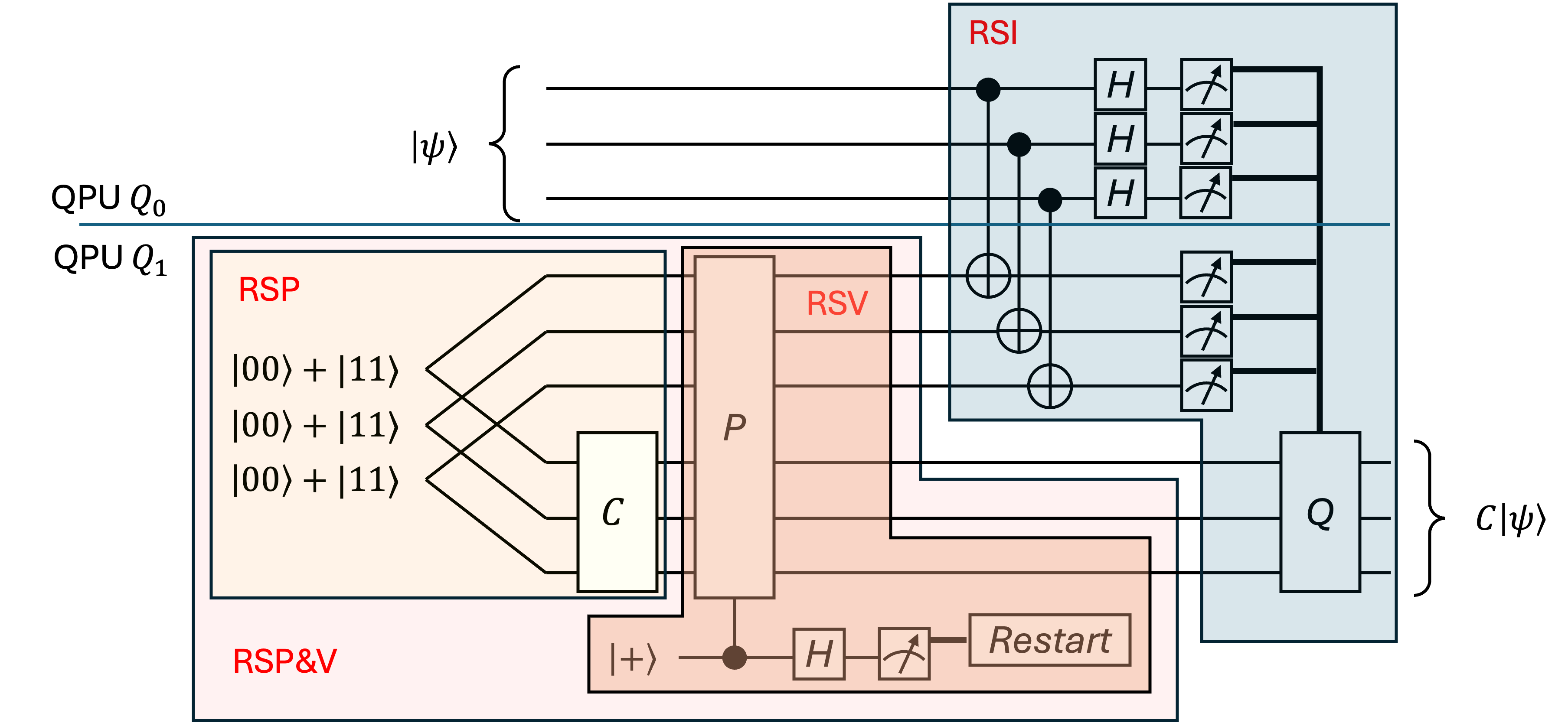}

    \vspace{10pt}

    (b) 
    
    \includegraphics[width=0.8\linewidth]{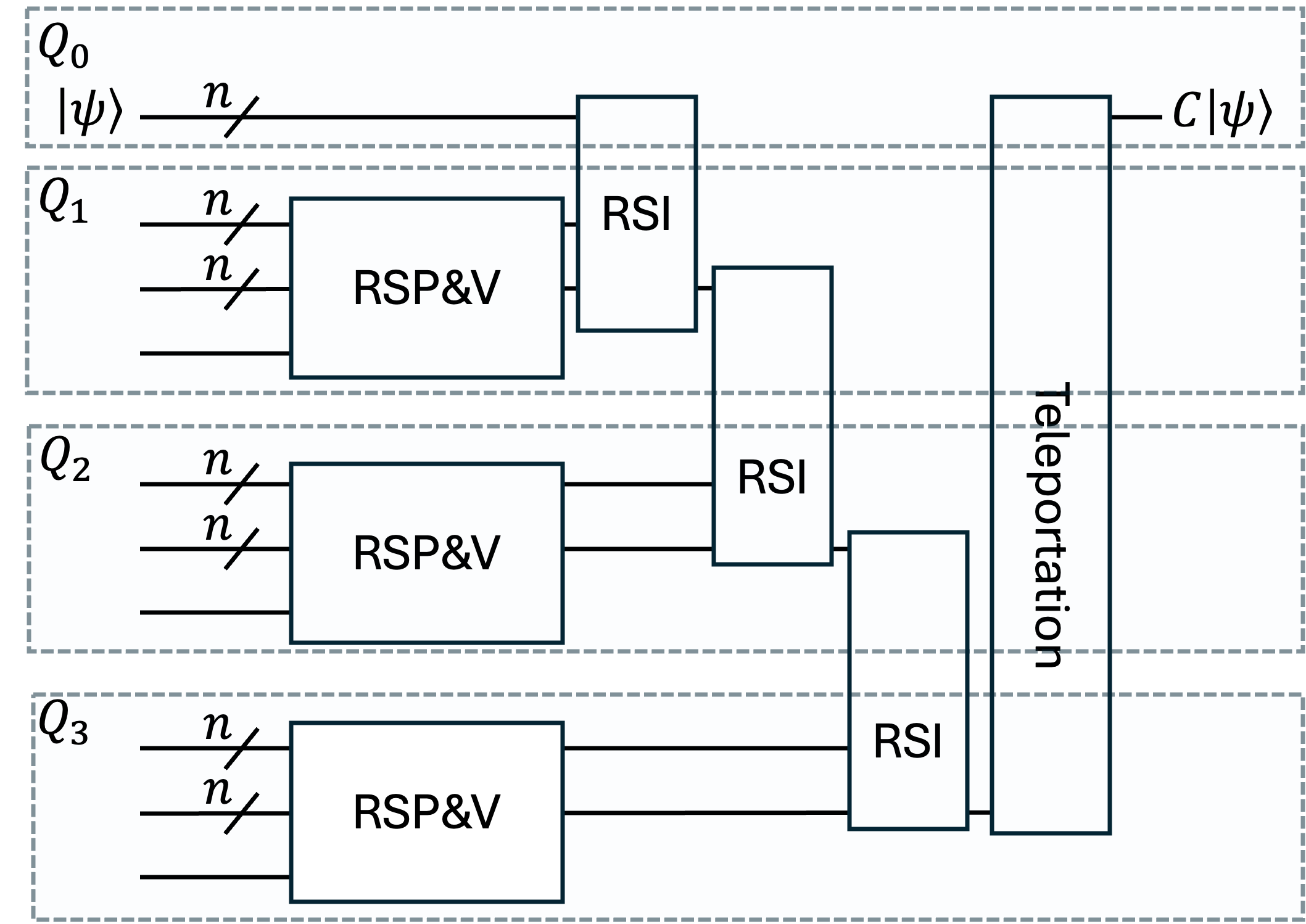}

    \caption{Distributed $\CliNR$. (a) The three components of $\CliNR$: RSP, RSV, and RSI. On a monolithic architecture, all three would be executed on a single $3n+1$ qubit QPU.  In the distributed architecture, RSP and RSV (\RSPV) would be executed on one QPU   requiring $2n+1$ computational qubits and RSI would require remote gates. (b) Distributed $\CliNR$ with $t=3$ (4 QPUs). The \RSPV stages which are expected to have a relatively large depth are executed in parallel. For simplicity, qubit rails are terminated when they are no longer used.}
    \label{fig:clinr}
\end{figure}

The $\CliNR_{1,r}(C)$ implementation of an $n$-qubit Clifford circuit $C$ of depth $\depth$ follows three stages (see \cref{fig:clinr}):
\begin{enumerate}
    \item {\bf Resource state preparation (RSP)}. Preparation of a $2n$-qubit resource state $\ket{\varphi}$. Starting with the preparation of  $2n$ Bell pairs and continuing with the implementation of $C$ on $n$ qubits, one from each Bell pair. The depth of RSP is $\depth+3$ where the term $3$ comes from the Bell state preparations.

    \item {\bf Resource state verification (RSV)}. The resource state is verified using  $r$ stabilizer measurements.  If one of these measurements returns a non-trivial outcome, the verification fails and stops. Then, the resource state is discarded and the process restarts from RSP. 
    The depth of RSV depends on the number of restarts triggered by the stabilizer measurements.  
    
    \item {\bf Resource state injection (RSI)}. If no errors are detected during RSV, the subcircuit $C$ is applied to the $n$-qubit input state $\ket{\psi}$ by consuming the verified resource state $\ket{\varphi}$. This can be done in depth 4 as shown in \cref{fig:clinr}(a).    
\end{enumerate}

If $C$ is large, then the expected number of restarts might be large. In such cases, it is advantageous to break $C$ into $t$ subcircuits $C_k$ and implement $\CliNR_{1,r}(C_k)$ on each one individually. These can be chained to provide an implementation of $C$ which we denote $\CliNR_{t,r}(C)$.

\subsection{Distributed CliNR}
\label{subsec:distributed_clinr}

Here, we propose a distributed version of the $\CliNR$ scheme (see \cref{fig:clinr} (b)). As with standard $\CliNR$, we consider implementing a $n$-qubit Clifford circuit $C$ on an input state $\ket{\psi}$. The implementation of distributed $\CliNR_{t,r}$ requires $t+1$ QPUs denoted $Q_k$ with $n_c=2n+1$ compute qubits and $n_s=2n$ storage qubits each. It begins with dividing $C$ into $t$ subcircuits $C_k$ of equal depth (as with monolithic $\CliNR_{t,r}$).  Then, starting at the beginning of the circuit, we do the following:

\begin{enumerate}
    \item {\bf Parallel \RSPV}. For each $C_k$, implement the \RSPV blocks of $\CliNR_{1,r}(C_k)$ on $Q_k$.  

    \item {\bf Serial injection}. We say that  $Q_k$, is \emph{ready for injection} if the following two conditions are met: 
      \begin{enumerate}
        \item The implementation of \RSPV completed on $Q_k$. 
        \item If $k>1$, RSI from $Q_{k-2}$ to $Q_{k-1}$ completed. 
    \end{enumerate}
    When a QPU is ready for injection we inject $C_k$ to the state $\prod_{\ell = 1}^{k-1}C_\ell\ket{\psi}$ on $Q_{k-1}$ by consuming the resource state as in the RSI step of the monolithic $\CliNR$ scheme. However, the $CX$ are now remote $CX$ gates that consume a Bell state.  Note that the qubits storing the new state $\prod_{\ell = 1}^{k}C_\ell\ket{\psi}$ are now in~$Q_k$.  
    
    If the number of stored Bell states is not sufficient to complete all the remote operations, it can start with the available Bell pairs and consume the remaining Bell pairs as they become available. RSI is complete only when all $n$ Bell pairs have been consumed. 

    If all Bell pairs are available when RSI on $Q_k$ starts, then the process has depth $4$. If fewer Bell pairs are available at the start, then the process has depth up to $n\tau_e+4$. 

    \item{\bf Teleportation}.  Teleport the $n$-qubit state $C\ket{\psi}$ from $Q_t$ to $Q_0$. 
\end{enumerate}

\subsection{Three implementations of a Clifford circuit}
\label{subsec:three_implementations}

In what follows, we consider three implementations of a $n$-qubit Clifford circuit with size $s$.
The {\em direct implementation} executes the $s$ gates of the circuit on a monolithic architecture with a $n$-qubit QPU.
The {\em monolithic $\CliNR_{t, r}$ implementation} uses a single QPU with $3n+1$ qubits.
The {\em distributed $\CliNR_{t, r}$ implementation} uses a distributed architecture with $t+1$ QPUs with $2n+1$ compute qubits and $2n$ storage qubits each.

\section{Simulation results}
\label{sec:simulation}

\begin{figure} 
    \centering
    (a)
    
    \includegraphics[width=0.99\linewidth]{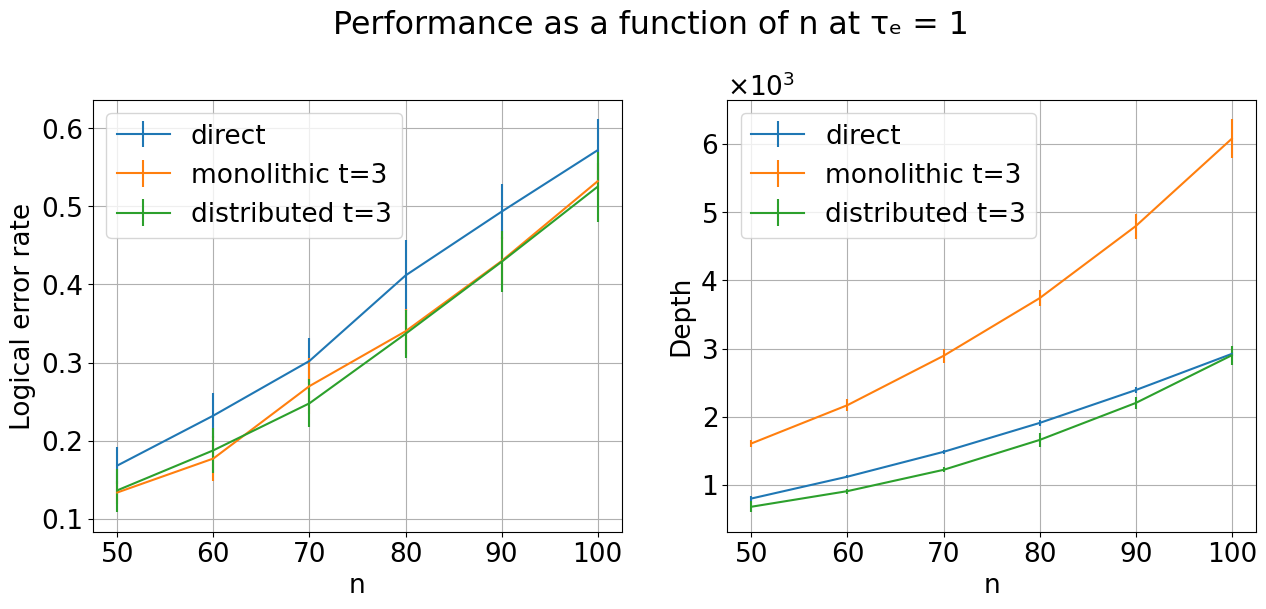}

    \vspace{10pt}
    
    (b)
        \includegraphics[width=0.99\linewidth]{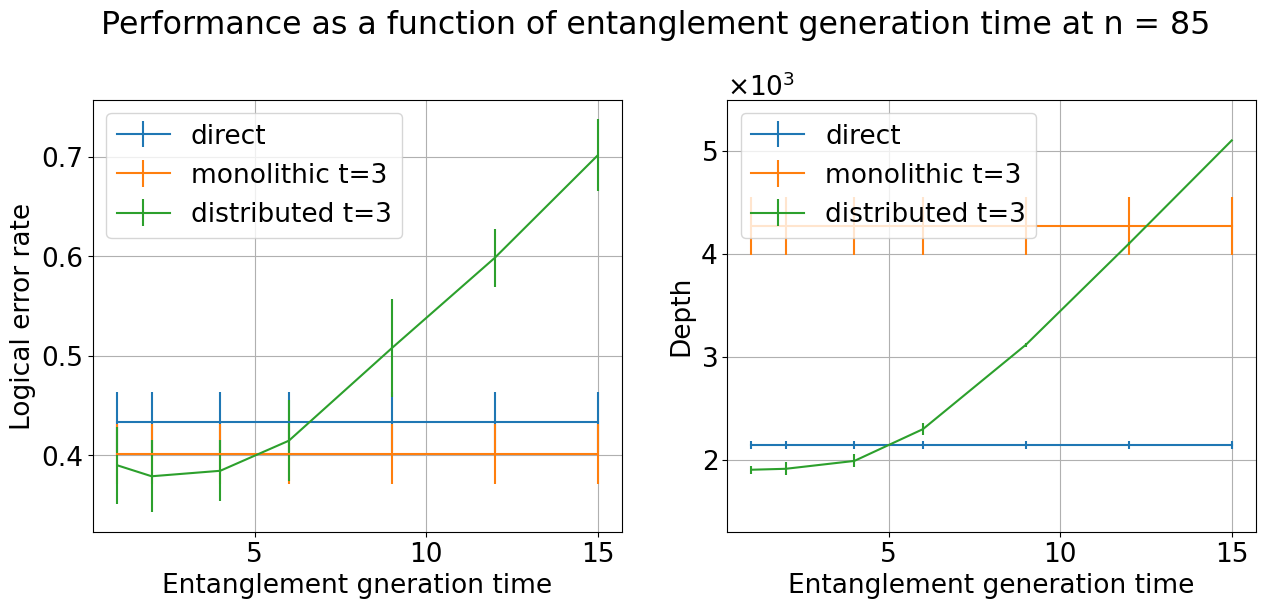}

    \caption{Simulation results for the three different implementation, direct, monolithic $\CliNR_{3,3}$ and distributed $\CliNR_{3,3}$. (a) Logical error rate (left) and depth (right) as a function of $n$. (b) Logical error rate (left) and depth (right) as a function of entanglement generation time at $n=85$.}
    \label{fig:simulationresults}
\end{figure}

Here, we present numerical results comparing the direct implementation with the monolithic and distributed CliNR implementations of random Clifford circuits.
The circuit is split into three subcircuits with equal depth $\pm1$ and we simulate the monolithic and distributed $\CliNR_{3,3}$ implementations.

Noise is applied to gates, measurements and idle qubits. We use a circuit level noise model where each operation is followed by a depolarizing channel applied to the support of the operation. The noise rate for intra-QPU one-qubit and two-qubit operations is $p/10$ and $p$ respectively.   
Remote two-qubit gates have a noise rate of $3p$ to account for the consumption of a Bell pair. That means that each intra-QPU two qubit gate, for example, is followed by the identity with probability $1-p$ or one of the 15 non-trivial Pauli errors with probability $p/15$ each. Furthermore, idle qubits in the compute module experience a depolarizing channel with noise rate $p/100$ (per depth 1 layer). Classical measurement outcomes are flipped with probability $p/10$.  For deriving bounds on depth (\cref{sec:Theory result}) we use a uniform noise model where all of the above noise rates are $p$ and there is no idle noise. 

We consider the regime of $n=50$ to $n=100$, with particular focus on $n=85$ where a monolithic $\CliNR$ implementation requires 256 qubits, and we set the noise parameter to $p=10^{-4}$, which was demonstrated recently in ion trap quantum computers \cite{hughes2025trapped}.

We furthermore choose $C$ to be a random Clifford with size $n^2$ (achieved by either truncating or padding with random Clifford gates as in \cite{brodutch2025recursive}). Each data point with fixed $n$ and $\tau_e$ in \cref{fig:simulationresults} is produced by running a Stim \cite{Gidney_Stim_a_fast_2021} simulation with 20 different circuits  $C$, and 200 samples are taken with each circuit to produce logical error rate and average depth. Error bars are taken as the standard deviation over the 20 different Clifford circuits. 

In \cref{fig:simulationresults} (a), we fix $\tau_e = 1$ and vary the number of qubits from $n=50$ to $n=100$. 
We observe that distributed $\CliNR$ reduces the logical error rate as much as the monolithic $\CliNR$ but it does so with a significantly reduced depth. Indeed, the depth of distributed $\CliNR$ is even lower than the depth of the direct implementation.
We expect that larger values of $t$ further increase the depth advantage of distributed CliNR, in particular in the regime of large $n$.

In \cref{fig:simulationresults} (b), we keep $n=85$ and vary over $\tau_e$. For all $\tau_e \leq 5$, we see that distributed $\CliNR$ has lower depth and better logical error rates than both the direct implementation and monolithic $\CliNR$. As $\tau_e$ grows, the impact of generating entanglement grows, increasing both the depth and the logical error rate.

\section{Asymptotic requirements for entanglement production}
\label{sec:Theory result}

In this section, we establish bounds on the depth of monolithic and distributed CliNR and we prove that a small number of parallel links is sufficient to avoid stalling distributed CliNR when the required Bell pairs are not available in time.

Consider a Clifford circuit $C$ split into $t$ subcircuits $C_i$ with $i=1,\dots,t$.
Denote by $\delta_i$ the depth of the \RSPV for $C_i$ in the absence of restart and let $q_i$ be the restart probability for this subcircuit.
Note that $q_i$ is the same in both the monolithic and the distributed CliNR implementation because the \RSPV subroutine is performed in the exact same way in both implementations.

Throughout this section, we simplify the protocol by assuming that the \rsv does not trigger a restart as soon as a non-trivial stabilizer measurement is observed but only at the end of the $r$ stabilizer measurements.
This makes the proof simpler without significantly increasing the circuit depth because $r$ is typically small.

\begin{proposition} [monolithic CliNR depth]
The expected depth of the monolithic CliNR implementation of $C$ is 
\begin{equation}
    \sum_{i=1}^t \frac{\delta_i}{1-q_i} + 4t \cdot
\end{equation}
\end{proposition}

\begin{proof}
The expected depth of \RSPV for subcircuit $C_i$ (including restarts) is given by 
\begin{equation}
    \sum_{j \geq 0} j \delta_i q_i^{j-1} (1-q_i)
    = \frac{\delta_i}{1-q_i} \cdot
\end{equation}
This leads to the result by linearity of the expectation and adding the contribution $4t$ for the $t$ rounds of RSI.
\end{proof}

\begin{proposition} [distributed CliNR depth]
The expected depth of the distributed CliNR implementation of $C$ is upper bounded by
\begin{equation}
    \max\left(
        \frac{\max_i(\delta_i) (\ln t + 3)}{\min_i(1-q_i)}, 
        n\tau_e
    \right) 
    + 4t
\end{equation}
\end{proposition}

\begin{proof}
The parallel \RSPV terminates once the \RSPV subroutine succeeds for all the subcircuits $C_i$.
Based on \cref{lemma:expected_stopping_time}, the expected number of \rsp per subcircuit before all the $C_i$ pass the verification is upper bounded by
\begin{align}
    \frac{\ln t + 3}{\min_i(1-q_i)} \cdot
\end{align}
Each attempt contributes to the depth by at most $\max_i(\delta_i)$.
Therefore, the expected depth of the parallel \RSPV is upper bounded by 
\begin{align}
    \frac{\max_i(\delta_i) (\ln t + 3)}{\min_i(1-q_i)} \cdot
\end{align}
The $t-1$ rounds of RSI and the final teleportation add $4t$ to the depth.
However, for these to be executed, $n$ Bell pairs must be available for each pair of QPUs at the end of the parallel \RSPV.
Generating these pairs takes depth $n \tau_e$, which leads to the result.
\end{proof}

If all subcircuits have the same depth $\delta = \delta_i$ and the same restart probability $q=q_i$, the monolithic CliNR depth is 
\begin{align}  \label{eq:depth mono}
\frac{t \delta}{1-q} + 4t
\end{align}
and the distributed CliNR depth is at most
\begin{align} \label{eq:dist_bound}
\max\left(
    \frac{\delta}{1-q} (\ln t + 3), 
    n \tau_e
\right)
+ 4t \cdot
\end{align}
Then, if
\begin{equation} \label{eq:entanglement_generation_requirement}
n \tau_e \leq \frac{\delta}{1-q} (\ln t + 3)
\end{equation}
the entanglement production is fast enough not to incur any increase of the average runtime of distributed CliNR.
In this regime, comparing the first terms in \cref{eq:depth mono} and \cref{eq:dist_bound}, we see that the \RSPV depth drops from $O(t)$ to $O(\ln t)$, showing that the depth of distributed $\CliNR$  can be shorter than monolithic $\CliNR$.

For \cref{eq:entanglement_generation_requirement} to be satisfied
for linear-depth $n$-qubit circuits with $n \rightarrow +\infty$, it suffices to have $\tau_e \geq A \frac{\ln t}{t}$ for some constant $A$. This can be achieved by relaxing assumption (2) and 
using  $O\left(\frac{t}{\ln t}\right)$ links producing entanglement in parallel.
Remarkably, this number does not depend on the number of qubits inside each module.

\section{Conclusions}
\label{sec:Conclusions}

We identify a promising application to demonstrate the capability of a distributed quantum computer.
The fact that this application works with slow interconnects makes it a promising candidate for a near-term experiment.
Current proposals could lead to further speed up of entanglement generation through memory-enhancement~\cite{bhaskar2020experimental}, by increasing the photon collection rate using a cavity~\cite{cirac1997quantum, stute2012tunable, casabone2013heralded, krutyanskiy2019light, krutyanskiy2023entanglement, walker2021single} or through electron juggling~\cite{moore2025electron} and may achieve high fidelity without distillation~\cite{saha2025high} which could further increase the gap between distributed CliNR and its monolithic counterparts.

Moreover, our distributed CliNR scheme can be applied to implement the Clifford part of the conjugated Clifford circuit proposed in \cite{bouland2017complexity}.
This is a potential avenue for a demonstration of quantum superiority in a distributed setting.
Moreover, the error correction capability of CliNR can potentially increase the output fidelity of the experiment increasing the gap with the classical simulation.

One can consider other ways to reduce the CliNR depth. First, one could break the circuit into two subcircuits and implement them on both sides of the entangled pair during RSP, which could reduce the depth of \RSPV by a factor up to $2$. A second method is to run $m$ instances of $\CliNR_{1,r}(C)$ in parallel and inject the first one that succeeds. This can reduce the depth by a factor up to $1/m$.

\section{Acknowledgment}
The authors thank Johannes Borregaard, Daniel Moore, Edwin Tham, Felix Tripier, Min Ye, John Gamble and the whole IonQ team for insightful discussions.

\bibliography{references}

\appendix

\section{Technical lemma}

The following lemma, which is a variant to the coupon collector problem, is used to establish an upper bound on the expected time for the parallel \RSPV in distributed CliNR.

\begin{lemma} \label{lemma:expected_stopping_time}
Let $X_1, \dots, X_t$ be $t$ Bernoulli random variables such that $\Prob(X_i=1) = p_i$.
Let $T_i$ be the occurrence of the first sample with value $1$ from $X_i$ and let $T = \max_i(T_i)$.
Then, we have
\begin{equation}
    \Expectation(T) \leq \frac{\ln t + 3}{p_0}
\end{equation}
where $p_0 := \min_i(p_i)$.
\end{lemma}

\begin{proof}
For each $T_i$, we have
\begin{equation}
    \Prob(T_i > k) = (1-p_i)^k \leq e^{-kp_i} \leq e^{-kp_0} \cdot
\end{equation}
Therefore, a union bound yields
\begin{equation}
    \Prob(T > k) \leq \sum_{i=1}^t \Prob(T_i > k) \leq t e^{-kp_0} \cdot
\end{equation}
To bound the expectation of $T$, we use the tail sum formula and we split the sum at $s := \lfloor \ln t/p_0 \rfloor$ which gives
\begin{align}
    \Expectation(T) 
    & = \sum_{k \geq 0} \Prob(T > k) \\
    & \leq \sum_{k = 0}^{s} 1 
        + \sum_{k > s} \Prob(T > k) \\
    & \leq \frac{\ln t}{p_0} + 1 + \frac{1}{1-e^{-p_0}} \cdot
\end{align}
The last term is upper bounded by $2/p_0$ which gives
$
\Expectation(T) \leq \frac{\ln t + 3}{p_0}
$, proving the lemma.
\end{proof}

\end{document}